\documentclass[submission,copyright,creativecommons]{eptcs}
\providecommand{\event}{TARK 2025} 

\usepackage{iftex}

\usepackage[pdftex]{graphicx}
\usepackage[table]{xcolor}
\usepackage{float}
\usepackage{hhline}
\usepackage{amssymb}
\usepackage{amsmath}
\usepackage{amsthm}
\usepackage{eurosym}
\usepackage{hyperref}
\usepackage{hhline}
\usepackage{tikz}
\usetikzlibrary{positioning}
\usepackage{tkz-euclide}

\newtheorem{proposition}{Proposition}

\newtheorem{lemma}{Lemma}
\newtheorem{theorem}{Theorem}
\newtheorem{corollary}{Corollary}
\newcommand{\commentout}[1]{}

\ifpdf
  \usepackage{underscore}         
  \usepackage[T1]{fontenc}        
\else
  \usepackage{breakurl}           
\fi

\title{Optimal Information Design \\ in Sender-Receiver Cheap Talk Interactions}
\author{Itai Arieli
\institute{Technion\\ Haifa, Israel}
\email{iarieli@technion.ac.il}
\and
Ivan Geffner
\institute{Utrecht University\\
Utrecht, Netherlands}
\email{i.e.geffnerfuenmayor@uu.nl}
\and
Moshe Tennenholtz
\institute{Technion\\ Haifa, Israel}
\email{moshet@technion.ac.il}
}
\def\titlerunning{Optimal Information Design}
\def\authorrunning{I. Arieli, I. Geffner \& M. Tennenholtz}
\begin{document}
\maketitle

\begin{abstract}
This paper considers the dynamics of cheap talk interactions between an oblivious receiver and a sender with different amounts of information. Even though it may seem that having additional information about the state of the game is always beneficial to the sender, we show that there are cases in which garbling the information of a fully informed sender can improve not only receiver's utility in equilibrium, but also that of the sender herself. We also provide efficient algorithms that output the optimal amount of information in sender-receiver scenarios with binary actions and extend some of these results to settings with multiple senders and one receiver.
\end{abstract}

\section{Introduction}

 Consider a cheap talk interaction between a seller and a buyer in which the seller has in its possession an art piece that can be either an original (\textbf{OG}), a fake that is almost indistinguishable from an original (\textbf{IF}), or a fake easily distinguishable by an art expert (\textbf{DF}). The buyer wants to buy the piece only if it is an original. The seller wants to sell it only if it is of types \textbf{OG} or \textbf{IF}, since otherwise it is highly likely that she will eventually lose reputation for selling fakes (even if the buyer cannot tell right away). If the item is not sold, both agents get $0$ utility. Instead, if the buyer buys the item, the exact utilities are shown in the following table.

$$\begin{array}{|c||c|c|}
\hline
\mbox{Type} & \mbox{Seller} & \mbox{Buyer}\\
\hhline{|=#=|=|}
\mbox{\textbf{OG}} & 1 & 1 \\
\hline
\mbox{\textbf{IF}} & 1 & -5\\
\hline
\mbox{\textbf{DF}} & -5 & -5\\
\hline
\end{array}.$$

Suppose that, initially, the item is equally likely of being \textbf{OG}, \textbf{IF} or \textbf{DF}. Moreover, assume that the seller is an art expert, and thus can distinguish between types \textbf{OG} and \textbf{DF}, and also that she knows if the item is a fake or not from the source (which means that she can distinguish between all three types). Suppose that the buyer is an art collector, but not necessarily an art expert (i.e., the buyer has no information about the item's type besides the prior). It is easy to check that, in this scenario, the only Nash equilibrium is the one in which the buyer never buys the item, giving both agents $0$ utility. To see this, note that, if the seller could convince the buyer to buy the item in some cases, she would definitely do so whenever the item is of type \textbf{IF}, giving the buyer a strictly negative expected utility.

Consider an identical scenario but in which the seller is \textbf{not} an art expert. This means that the seller still knows if the piece is a fake or not from the source but, whenever the item is fake, she can't tell how good of a fake it is (i.e., she can't tell if it is of type \textbf{IF} or \textbf{DF}.) Here, there is a Nash equilibrium that gives $\frac{1}{3}$ expected utility to each agent in which the seller signals if the item is an original or not, and the buyer buys it only if it is original (note that the sender does not want to sell the item whenever it is fake). Thus, in this example, the receiver is better off by not being knowledgeable about art. This raises the natural question of, given a general sender-receiver setting, determining the \emph{optimal} amount of knowledge that the sender must have in order for her (resp., the receiver) to maximize her expected utility in equilibrium.  More precisely, we start with a setting in which the sender has full information about the state of the game, and our aim is constructing an information filter that garbles the information that is disclosed to the sender in a way that the best equilibrium of the resulting game gives her (resp., the receiver) the maximum possible utility. In this paper we focus on constructing such filter in a setting with binary actions. In particular, we provide $O(k \log k)$ algorithms that output the optimal disclosure of information to the sender, where $k$ is the number of states. If we restrict our attention to the best possible equilibrium for the receiver, we also provide an efficient algorithm when there are multiple senders that are equally knowledgeable. A perhaps surprising property of our setting is that, in the example above, the fact that both players can simultaneously increase their utility is no coincidence. In fact, we prove in our analysis that, if the sender can increase her utility by garbling her information, the utility of the receiver increases as well (see Proposition~\ref{prop:pareto}). 

\subsection{Related Literature}

Restricting the information available to the sender has been increasingly gaining traction in the community. Most notably, Bergemann, Brooks and Morris~\cite{bergemann2015limits} considered a buyer-seller setting and studied possible ways to limit the seller's information. In their work, they characterized the set of all pairs of possible utilities achievable in equilibrium. In particular, they provide an algorithm that achieves the optimal way to limit the seller's information to maximize the expected buyers' utility. Ichihashi~\cite{ichihashi2019limiting} considered a Bayesian persuasion setting and studied how the outcome of the interaction is affected when the sender's information is restricted. One of their results is that, if the receiver restricts sender information in a pre-play stage, the best utility that the receiver can get in this setting coincides with the one that the receiver would get in the ``flipped game'', where the receiver persuades the sender. 
We study a similar problem, except that in our case agents have no commitment power.
Other papers have studied settings with limited communication between the sender and the receiver both in the context of common-interest coordination games~\cite{blume2013language, de2003game} and cheap talk models~\cite{jager2011voronoi, hagenbach2020cheap}. Restricting the  information available also became quite relevant in the context of moderating large language model models in order for them to avoid producing undesirable responses in some prompts~\cite{patil2023can,liu2024protecting,glukhov2023llm}. Our underlying model is also similar to that of Bayesian persuasion~\cite{kamenica2011bayesian}, especially when the sender has no commitment power~\cite{lipnowski2020cheap, corrao2023mediated, arieli2023mediated, arieli2023resilient}.

The rest of the paper is organized as follows. In Section~\ref{sec:model} we introduce the main concepts and definitions used throughout the paper. 
In Section~\ref{sec:results-filter} we state our main results regarding the computation of the best filters in sender-receiver games with one and two senders. These results are later proved in Sections~\ref{sec:proof-1-sender} and \ref{sec:proof-2-senders}. We end with a conclusion in Section~\ref{sec:conclusion}.

\section{Basic Definitions}\label{sec:model}

\subsection{Information Transmission Games}

An \emph{information transmission game} involves a \emph{sender} $s$ and a \emph{receiver} $r$, and is defined by a tuple $\Gamma = (A, \Omega, p, M, u)$, where $A = \{a_1, \ldots, a_\ell\}$ is the set of actions, $\Omega = \{\omega_1, \ldots, \omega_n\}$ is the set of possible states, $p$ is a commonly known prior distribution over $\Omega$ that assigns a strictly positive probability to each possible state, $M$ is a finite set that contains the messages that the sender can send to the receiver 
($M$ is usually assumed to be finite),
and $u : \{s, r\} \times \Omega \times A \longrightarrow \mathbb{R}$ is a utility function such that $u(i, \omega, a)$ gives the utility of player $i$ (where $i$ is either the sender or the receiver) when action $a$ is played at state $\omega$. Each information transmission game instance is divided into three phases. 
\begin{itemize}
    \item \textbf{Phase 1:} A state $\omega \in \Omega$ is sampled according to the distribution $p$ and is disclosed to the sender.
    \item \textbf{Phase 2:} The sender sends a message $m \in M$ to the receiver.
    \item \textbf{Phase 3:} The receiver plays an action $a \in A$ and each player $i$ receives $u(i, \omega, a)$ utility.
\end{itemize}

Given an information transmission game $\Gamma = (A, \Omega, p, M, u)$, a strategy profile $\vec{\sigma}$ for $\Gamma$ consists of a pair of strategies $(\sigma_s, \sigma_r)$ for the sender and the receiver, where $\sigma_s$ is a map from $\Omega$ to distributions over $M$ (which is denoted by $\Delta(M)$), and $\sigma_r$ is a map from $M$ to $\Delta(A)$.
We say that $\vec{\sigma}$ is a Nash equilibrium if no player can increase its utility by defecting from $\vec{\sigma}$. More precisely, if we denote by $u_i(\vec{\sigma})$ the expected utility that $i$ gets when players play $\vec{\sigma}$, then $\vec{\sigma}$ is a Nash equilibrium if $u_i(\vec{\sigma}_{-i}, \tau_i) \le u_i(\vec{\sigma})$ for all players $i \in \{s, r\}$ and all strategies $\tau_i$ for $i$.

\subsection{Information Aggregation Games}

At a high level, \emph{information aggregation games} are information transmission games with multiple senders. For a rigorous definition, an information aggregation game consists of a receiver $r$, and a tuple $\Gamma = (S, A, \Omega, p, M, u)$ where $S = \{s_1,s_2,\ldots,s_k\}$ is the set of senders and $A = \{a_1, \ldots, a_\ell\}$, $\Omega = \{\omega_1, \ldots, \omega_k\}$, $p \in \Delta(\Omega)$, $M$ and $u$ are defined as in information transmission games, with the only exception that $u$ is a function from $S \cup \{r\} \times \Omega \times A$ to $\mathbb{R}$ instead of a function from $\{s, r\} \times \Omega \times A$ to $\mathbb{R}$. An information aggregation game runs in the same way as an information transmission game, except that in phase $1$ the state is disclosed to all senders, and in phase 2 each sender $s_i$ sends a message $m_i$ to the receiver.

\subsection{Filtered Information Transmission Games}

We are interested in settings where the sender does not have a full picture about the state of the game. In order to model this, we will assume that, in Phase 1, the state $\omega$ is not directly disclosed to the sender, but that instead it goes through an \emph{information filter} $X: \Omega \to \Delta(\{0,1\}^*)$ that maps each state $\omega$ to a distribution over possible signals, and these signals are binary strings of arbitrary length (we use the notation $\{0,1\}^* := \bigcup_{n \ge 0} \{0,1\}^n$). The sender receives a signal sampled from $X(\omega)$ instead of $\omega$ itself. If $\Gamma = (A, \Omega, p, M, u)$ is an information transmission game and $X: \Omega \to \Delta(\{0,1\}^*)$ is a filter, we denote by $\Gamma(X)$ the resulting filtered information game where the state goes through $X$ before being disclosed to the sender. Filtered information aggregation games are defined analogously. For future reference, given an information transmission game $\Gamma$, we define a \emph{sender-optimal filter} $X_s^\Gamma$ as a filter in which the utility that the sender gets in the best equilibrium is maximal. We define a \emph{receiver-optimal filter} $X_r^\Gamma$ analogously. 

\section{Main Results}\label{sec:results-filter}
We next address the question of finding the optimal amount of information the sender must have in information transmission games and information aggregation games with binary actions.
\begin{theorem}~\label{thm:binary-actions}
Let $\Gamma = (A, \Omega, p, M, u)$ be an information transmission game with $A = \{0,1\}$. Then, there exists an algorithm $\pi$ that outputs a receiver-optimal filter (resp., a sender-optimal filter). The algorithm runs in $O(k \log k)$ time, where $k = |\Omega|$.
\end{theorem}

If there are two senders instead, the following result shows that we can reduce the problem of computing a receiver-optimal filter to a linear programming instance where the number of variables and constraints are linear over the number of states.

\begin{theorem}~\label{thm:two-senders}
Let $\Gamma = (S, A, \Omega, p,M, u)$ be an information aggregation game with $|\Omega| = k$, $S = \{s_1, s_2\}$, and $A = \{0,1\}$. Then, finding a receiver-optimal filter reduces to a linear programming instance with $k$ variables and $2k+2$ constraints.
\end{theorem}

Note that, if we focus on the receiver, we only provide results for the cases of one and two senders. If $\Gamma = (S, A, \Omega, p, M, u)$ is an information aggregation game with $|S| \ge 3$ there is a trivial strategy profile $\vec{\sigma}_{MAX}$ that gives the receiver the maximal possible utility of the game with no filters. The senders simply forward the state of the game to the receiver, and the receiver plays the action that gives her the most utility on the state sent by a majority of the senders. It is easy to check that this is indeed a Nash equilibrium: the senders do not get any additional utility by defecting since the other senders will still send the true state (which means that the receiver will be able to compute the true state as well), and the receiver does not get any additional utility either since she is always playing the optimal action for each possible state. This gives the following result.

\begin{theorem}\label{thm:majority}
Let $\Gamma = (A, \Omega, p,M, u)$ be an information aggregation game with $|S| \ge 3$. Then, $\vec{\sigma}_{MAX}$ is a Nash equilibrium that gives the receiver the maximal possible utility of the game.
\end{theorem}

In the following sections, we give algorithms that output the best Nash equilibria for the settings described in Theorems~\ref{thm:binary-actions} and \ref{thm:two-senders}, respectively.

\section{Proof of Theorem~\ref{thm:binary-actions}}\label{sec:proof-1-sender}

In this section we provide an algorithm that outputs the receiver-optimal filter for any information transmission game $\Gamma^{F} = (A, \Omega, p, M, u)$ with $S = \{s\}$ and $A = \{0,1\}$. The sender-optimal filter can be computed analogously (see Section~\ref{sec:sender-algorithm} for more details). We start with the following proposition:

\begin{proposition}\label{prop:characterization1}
Let $\Gamma = (A, \Omega, p, M, u)$ be an information transmission game and let $X: \Omega \to \Delta(\{0,1\}^*)$ be an information filter. Then, there exists a Pareto-optimal Nash equilibrium $\vec{\sigma}$ in $\Gamma(X)$ such that, in $\vec{\sigma}$, either:
\begin{itemize}
    \item The receiver always plays 0.
    \item The receiver always plays 1.
    \item The receiver always plays the best action for the sender.
\end{itemize}
\end{proposition}

Before proving Proposition~\ref{prop:characterization1} we need additional notation. 
First, given a filtered information transmission game $\Gamma(X)$, let $u_i(x, a)$ be the expected utility of player $i$ on signal $x$ and action $a$. This expected utility can be computed with the following equation:
$$u_i(x,a) = \sum_{\omega \in \Omega} \Pr[\omega \mid x] \cdot u_i(\omega, a),$$
where $\Pr[\omega \mid x]$ is the probability that the realized state is $\omega$ conditional on the fact that the sender received signal $x$. Similarly, $\Pr[\omega \mid x]$ has the following expression:

$$
\Pr[\omega \mid x] = \frac{\Pr[\omega \longleftarrow p, \ x \longleftarrow X(\omega)]}{\sum_{\omega \in \Omega} \Pr[\omega \longleftarrow p,\  x \longleftarrow X(\omega)]}.
$$

Let $X_0$ be the set of signals $x$ in $\{X(\omega)\}_{\omega \in \Omega}$ such that 
$u_s(x, 0) > u_s(x, 1)$ (i.e., the set of signals in which the sender prefers $0$), let $X_1$ be the set of signals such that 
$u_s(x, 0) < u_s(x, 1)$, and $X_=$ be the set signals such that 
$u_s(x, 0) = u_s(x, 1)$. 
Given a strategy profile $\vec{\sigma}$ for $\Gamma(X)$, denote by $\sigma_s(x, m)$ the probability that the sender sends message $m$ given signal $x$ and denote by $\sigma_r(a, m)$ the probability that the receiver plays action $a$ given message $m$. Moreover, let $M_0^{\vec{\sigma}}$ denote the set of messages that have a strictly positive probability to be sent by the sender on at least one signal $x$ in which the sender prefers 0 (i.e., $M_0^{\vec{\sigma}}$ is the set of messages $m$ such that there exists $x \in X_0$ such that $u_s(x, m) > 0$). We define $M_1^{\vec{\sigma}}$ and $M_=^{\vec{\sigma}}$ analogously.

With this notation, the following lemma describes all strategy profiles in $\Gamma(X)$ that are incentive-compatible for the sender.

\begin{lemma}\label{lemma:char}
A strategy profile $\vec{\sigma}$ for $\Gamma(X)$ is incentive-compatible for the sender if and only if the following is satisfied:
\begin{itemize}
\item [(a)] If $m \in M^{\vec{\sigma}}_0$, then $\sigma_r(0, m) \ge \sigma_r(0, m')$ for all messages $m'$.
\item [(b)] If $m \in M^{\vec{\sigma}}_1$, then $\sigma_r(1, m) \ge \sigma_r(1, m')$ for all messages $m'$.
\end{itemize}
\end{lemma}

Lemma~\ref{lemma:char} states that, for a strategy profile to be incentive-compatible for the sender, the receiver should play $0$ with maximal probability on all messages that could be sent on signals in which the sender prefers 0, and the receiver should play $0$ with minimal probability on all messages that could be sent on signals in which the sender prefers 1. In particular, we have the following Corollary.

\begin{corollary}\label{col:char}
A strategy profile $\vec{\sigma}$ for $\Gamma(X)$ is incentive-compatible for the sender if and only if there exist two real numbers $\ell_{\min},\ell_{\max} \in [0,1]$ with $\ell_{\min} \le \ell_{\max}$ such that:
\begin{itemize}
\item [(a)]  $m \in M^{\vec{\sigma}}_0 \ \Longrightarrow \ \sigma_r(0, m) = \ell_{\max}$.
\item [(b)] $m \in M^{\vec{\sigma}}_1 \ \Longrightarrow \ \sigma_r(0, m) = \ell_{\min}$.
\item [(c)] $m \in M^{\vec{\sigma}}_= \ \Longrightarrow \ \ell_{\min} \le \sigma_r(0, m) \le \ell_{\max}$.
\end{itemize}
\end{corollary}

\begin{proof}[Proof of Lemma~\ref{lemma:char}]
Clearly, if (a) and (b) are satisfied, then $\vec{\sigma}$ is incentive-compatible for the sender. Conversely, suppose that $\vec{\sigma}$ is incentive-compatible for the sender but it doesn't satisfy (a). This means that there exists a signal $x \in X_0$ and a message $m$ that satisfies $u_s(x, m) > 0$ and such that $u_r(0, m) < u_r(0, m')$ for some other message $m'$. Therefore, if the sender sends $m'$ instead of $m$ whenever it receives signal $x$, it could increase its expected utility. This contradicts the fact that $\vec{\sigma}$ is incentive-compatible for the sender. The proof of the case in which $\vec{\sigma}$ doesn't satisfy (b) is analogous.
\end{proof}

Corollary~\ref{col:char} characterizes the necessary and sufficient conditions for a strategy profile $\vec{\sigma}$ in $\Gamma(X)$ to be incentive-compatible for the sender. We show next that Proposition~\ref{prop:characterization1} follows from adding the receiver incentive-compatibility constraints into the mix. Denote by $u_r^{\vec{\sigma}}(m, 0)$ the receiver's expected utility when playing action $0$ conditioned on the fact that it received message $m$ and that the sender plays $\sigma_s$. Then, we have the following cases:

\textbf{Case 1:} There exists $m \in M_0^{\vec{\sigma}}$ such that $u_r^{\vec{\sigma}}(m, 0) < u_r^{\vec{\sigma}}(m, 1)$: Since $\vec{\sigma}$ is incentive-compatible for the receiver, it must be the case that $\ell_{\min} = 1$, and therefore that $\sigma_r(0, m) = 1$ for all messages $m$.

\textbf{Case 2:} There exists $m \in M_1^{\vec{\sigma}}$ such that $u_r^{\vec{\sigma}}(m, 0) > u_r^{\vec{\sigma}}(m, 1)$: This time it must be the case that $\ell_{\max} = 0$, and therefore that $\sigma_r(0, m) = 0$ for all messages $m$.

\textbf{Case 3: } $u_r^{\vec{\sigma}}(m, 0) \ge u_r^{\vec{\sigma}}(m, 1)$ for all $m \in M_0^{\vec{\sigma}}$ and $u_r^{\vec{\sigma}}(m, 0) \le u_r^{\vec{\sigma}}(m, 1)$ for all $m \in M_1^{\vec{\sigma}}$: In this case, consider a strategy profile $\vec{\sigma}'$ that is identical to $\vec{\sigma}$ except that, whenever the receiver receives a message $m \in  M_0^{\vec{\sigma}}$, it plays action $0$ with probability 1 (as opposed to probability $\ell_{\max}$), and when the receiver receives a message $m \in M_1^{\vec{\sigma}}$, it plays action $1$ with probability 1 (as opposed to $1-\ell_{\min}$). It is easy to check that, by construction, $\vec{\sigma}'$ is a Nash equilibrium that Pareto-dominates $\vec{\sigma}$. Even so, $\vec{\sigma}'$ can be further improved: consider a strategy profile $\vec{\sigma}^s$ such that the sender sends message $0$ whenever she strictly prefers $0$ to $1$ or whenever she is indifferent and the receiver strictly prefers $0$ to $1$, and she sends message $1$ otherwise. Additionally, the receiver plays the action suggested by the receiver. It is easy to check that $\vec{\sigma}^s$ is a Nash equilibrium that Pareto-dominates $\vec{\sigma}'$ since the receiver plays the best action for the sender in both cases, but in $\vec{\sigma}^s$ she also plays the best action for herself whenever the sender is indifferent while she may not necessarily do so in $\vec{\sigma}'$.

This analysis provides a refinement of Proposition~\ref{prop:characterization1}. In fact, consider the following two strategy profiles:
\begin{itemize}
    \item Strategy $\vec{\sigma}^r$: Regardless of the signal received, the sender signals an empty $\bot$ message. The receiver plays the action that gives her the most utility with no information.
    \item Strategy $\vec{\sigma}^s$: After receiving the signal, the sender signals her preferred action to the receiver. The receiver plays the action sent by the sender.
\end{itemize}

We have the following characterization of Pareto-optimal Nash equilibria:

\begin{proposition}[Refinement of Proposition~\ref{prop:characterization1}]\label{prop:characterization2}
If $\vec{\sigma}^s$ is incentive-compatible for the receiver, then $\vec{\sigma}^s$ is a Pareto-optimal Nash equilibrium of $\Gamma(X)$. Otherwise, $\vec{\sigma}^r$ is a Pareto-optimal Nash equilibrium of $\Gamma(X)$.
\end{proposition}

Proposition~\ref{prop:characterization2} shows that Theorem~\ref{thm:binary-actions} reduces to find the filter $X$ such that the receiver (resp., the sender) maximizes her utility by playing $\vec{\sigma}^s$ (note that $\vec{\sigma}^r$ gives the same utility to both players independently of the filter applied). In the next sections we show how to efficiently compute such filters. However, before going into it, our analysis also shows the following.

\begin{proposition}\label{prop:pareto}
    If there exists a filter that strictly increases the sender's utility, it also (weakly) increases the receiver's utility.
\end{proposition}

\begin{proof}
If $\vec{\sigma}^s$ is a Nash equilibrium of $\Gamma$, it gives the maximum possible utility to the receiver in every single state. Therefore, if there exists a filter $X$ such that the sender gets more utility in $\Gamma(X)$ than in $X$, it must be that $\vec{\sigma}^r$ is a Pareto-optimal Nash equilibrium of $\Gamma$ while $\vec{\sigma}^s$ is a Pareto-optimal Nash equilibrium of $\Gamma(X)$. This lemma follows from the fact that, if $\vec{\sigma}^s$ is a Nash equilibrium of $\Gamma(X)$, it gives more utility to the receiver than $\vec{\sigma}^r$.
\end{proof}

\subsection{Computation of $X^\Gamma_r$}\label{sec:algorithm}

In this section we show how to compute a receiver-optimal filter. As discussed in Section~\ref{sec:proof-1-sender}, our aim is to find a filter $X^\Gamma_r$ such that $\vec{\sigma}^s$ is incentive-compatible for the receiver in $\Gamma(X^\Gamma_r)$ and such that the expected utility for the receiver with $\vec{\sigma}^s$ is as large as possible. We begin by showing an algorithm that runs in $O(k \log k)$ time, where $k$ is the number of states (i.e., the size of $\Omega$), and then we show the proof of correctness for the algorithm provided.

\subsubsection{An $O(k \log k)$ Algorithm}

For our algorithm, we restrict our search to \emph{binary} filters (i.e., filters that only send signals in $\{0,1\}$) and later, in Lemma~\ref{lem:simplification}, we show that this is done without loss of generality. With this restriction, we can describe possible filter candidates $X$ by a function $X_*$ that maps each state $\omega$ to the probability that $X(\omega)$ outputs $0$. Note that, without loss of generality, we can also restrict our search to filters in which the sender prefers action $0$ on signal $0$ and prefers action $1$ on signal $1$ (otherwise we can re-label the signals). If a filter satisfies this condition, we say that it is incentive-compatible for the sender. Analogously, if the receiver prefers action $0$ on signal $0$ and prefers action $1$ on signal $1$, we say that $X$ is incentive-compatible for the receiver. This shows that, without loss of generality, we can search for the binary filter that (a) is incentive-compatible for both the sender and the receiver, and (b) that gives the most utility for the receiver.

Before we start, let $\Omega^0$ (resp., $\Omega^1$) denote the set of states in which both the sender and the receiver prefer $0$ (resp., both prefer $1$) and let $\Omega_0^1$ (resp., $\Omega_1^0$) denote the set of states in which the sender strictly prefers $0$ and the receiver strictly prefers $1$ (resp., the sender strictly prefers $1$ and the receiver strictly prefers $0$). The following algorithm outputs the optimal filter $X^{OPT}$ for the receiver. 

\begin{enumerate}
\item \textbf{Step 1:} Set $(X_r^\Gamma)_*(\omega) = 1$ for all $\omega \in \Omega^0$.
\item \textbf{Step 2:} Set $(X_r^\Gamma)_*(\omega) = 0$ for all $\omega \in \Omega^1$.
\item \textbf{Step 3:} Sort all states $\omega \in \Omega_{0}^1 \cup \Omega_1^0$ according to the value $\frac{u_r(\omega, 0) - u_r(\omega, 1)}{u_s(\omega, 1) - u_s(\omega, 0)}$. Let $\omega^1, \omega^2, \ldots, \omega^{k'}$ be the resulting list.
\item \textbf{Step 4:} Set $(X_r^\Gamma)_*(\omega) = 1$ for all $\omega \in \Omega_1^0$ and $(X_r^\Gamma)_*(\omega) = 0$ for all $\omega \in \Omega_0^1$. If the resulting filter is incentive-compatible for the sender, output $X_r^\Gamma$.
\item \textbf{Step 5:} For $i = 1, 2, \ldots, k'$ do the following. If $(X_r^\Gamma)_*(\omega^i)$ is set to $1$ (resp., to $0$), set it to $0$ (resp., to 1). If the resulting filter is incentive-compatible for the sender, find the maximum (resp., the minimum) $q$ such that setting $(X_r^\Gamma)_*(\omega^i)$ to $q$ is incentive-compatible for the sender (we show in Section~\ref{sec:amortized} that it can be computed in amortized linear time). If the resulting filter  is incentive-compatible for the receiver, output $X_r^\Gamma$, otherwise output the constant filter $(X_r^\Gamma)_* \equiv 0$.
\end{enumerate}

It is important to note that the algorithm always terminates, since if we get to $k'$ in step 5, the resulting filter will always output the signal that the sender prefers. In Section~\ref{sec:amortized} we show how to compute $q$ and check the incentive-compatibility of the sender and of the receiver in amortized linear time, and in Section~\ref{sec:correctness} we show that the algorithm indeed produces the correct output.

\subsubsection{Checking Incentive-Compatibilities and Computing $q$ in Amortized Linear Time}\label{sec:amortized}

In this section we show how to check the players' incentive-compatibilities and how to compute $q$ in amortized linear time. For this purpose, let $d_i^s$ (resp., $d_i^r$) denote the difference in utility for the sender (resp., for the receiver) between playing $0$ and playing $1$. More precisely, $d_i^t := u_t(\omega^i, 0) - u_t(\omega^i, 1)$ for $t \in \{s,r\}$. Then, given a binary filter $X$, $\vec{\sigma}^s$ is incentive-compatible for the sender and the receiver in $\Gamma^F_X$ if and only if 
\begin{equation}\label{eq:ic-s}
\begin{array}{ll}
\sum_i p(\omega^i) \cdot d_i^t \cdot X_*(\omega^i) \ge 0\\
\sum_i p(\omega^i) \cdot d_i^t \cdot (1 - X_*(\omega^i)) \le 0
\end{array}
\end{equation}
for $t \in \{s, r\}$. The first and second equations state that the sender (resp., the receiver) prefers $0$ to signal $0$ and prefers $1$ on signal $1$, respectively.

It is straightforward to check that the incentive-compatibility equations for the sender monotonically increase and decrease, respectively, whenever we iterate in the fifth step of the algorithm. Therefore, we can simply pre-compute the following sums 
for $t \in \{s,r\}$, $b \in \{0,1\}$ and $i \in \{1,2,\ldots, k'\}$. 

$$\begin{array}{l}
Y_b^t := \sum_{\omega \in \Omega^b} p(\omega) \cdot (u_t(\omega, 0) - u_t(\omega,1))\\
W_{i,b}^{t} := \sum_{j < i, \omega^j \in \Omega_b^{1-b}} p(\omega^j) \cdot d_j^t\\
X_{i,b}^{t} := \sum_{j > i, \omega^j \in \Omega_b^{1-b}} p(\omega^j) \cdot d_j^t.
\end{array}$$

\commentout{
and their analogues for the receiver:

$$\begin{array}{l}
Y_0^r := \sum_{\omega \in \Omega_0^0} p(\omega) \cdot (u_r(\omega, 0) - u_r(\omega,1))\\
Y_1^r := \sum_{\omega \in \Omega_1^1} p(\omega) \cdot (u_r(\omega, 0) - u_r(\omega, 1))\\
W_{i,0}^{(r,s)} := \sum_{j < i, \omega^j \in \Omega_0^1} p(\omega^j) \cdot (u_r(\omega^j, 0) - u_r(\omega^j, 1))\\
W_{i,1}^{(r,s)} := \sum_{j < i, \omega^j \in \Omega_1^0} p(\omega^j) \cdot (u_r(\omega^j, 0) - u_r(\omega^j, 1))\\
W_{i,0}^{(r,r)} := \sum_{j > i, \omega^j \in \Omega_0^1} p(\omega^j) \cdot (u_r(\omega^j, 0) - u_r(\omega^j, 1))\\
W_{i,1}^{(r,r)} := \sum_{j > i, \omega^j \in \Omega_1^0} p(\omega^j) \cdot (u_r(\omega^j, 0) - u_r(\omega^j, 1)).
\end{array}$$
}

Note that these sums can be computed in amortized linear time over the total number of states. Once these sums are computed, if we are in the $i$th iteration of Step 5, to find $q$ we can check in constant time if the solution of any of the following linear equations is in $[0,1]$:

$$\begin{array}{l}
Y_0^s + W_{i,0}^{s} + p(\omega^i) \cdot q \cdot d_i^s + X_{i,0}^{s} = 0\\
Y_1^s + W_{i,1}^{s} + p(\omega^i) \cdot (1-q) \cdot d_i^s+ X_{i,1}^{s} = 0.\\
\end{array}$$

If there exists such a solution $q$, to check that it is incentive-compatible for both players, we should simply check if the equations in ~\ref{eq:ic-s} are satisfied, which can be rewritten as $Y_0^t + W_{i,0}^{t} + p(\omega^i) \cdot q \cdot d_i^t + X_{i,0}^{t} \ge 0$ and $Y_1^t + W_{i,1}^{t} + p(\omega^i) \cdot (1-q) \cdot d_i^t + X_{i,1}^{t} \le 0$, respectively. 
\commentout{
following inequalities are satisfied:

$$\begin{array}{l}
Y_0^s + W_{i,0}^{s} + p(\omega^i) \cdot q \cdot d_i^s + X_{i,0}^{s} \ge 0\\
Y_1^s + W_{i,1}^{s} + p(\omega^i) \cdot (1-q) \cdot d_i^s + X_{i,1}^{s} \le 0.\\
Y_0^r + W_{i,0}^{r} + p(\omega^i) \cdot q \cdot d_i^r + X_{i,0}^{r} \ge 0\\
Y_1^r + W_{i,1}^{r} + p(\omega^i) \cdot (1-q) \cdot d_i^r + X_{i,1}^{r} \le 0.\\
\end{array}$$
}
Note that all of these computations can be performed in constant time if the necessary sums are pre-computed. Also, the whole process runs in $O(k \log k)$ time since it takes $O(k \log k)$ operations to sort the elements of $\Omega_0^1 \cup \Omega_1^0$, $O(k)$ operations to compute the sums, and $O(1)$ operations in each iteration of Step 5. In the next section, we show that the algorithm's output is correct. 

\subsection{Proof of Correctness of $\pi$}\label{sec:correctness}

In this section, we show that the algorithm $\pi$ presented in Section~\ref{sec:algorithm} is correct. We begin by showing that we can indeed restrict our search to binary filters.

\begin{lemma}\label{lem:simplification}
Let $X$ be a filter such that $\vec{\sigma}^s$ is incentive-compatible for the receiver in $\Gamma(X)$. Then, there exists a binary filter $X'$ such that 
\begin{itemize}
\item[(a)] $\vec{\sigma}^s$ is incentive-compatible for the receiver in $\Gamma(X')$.
\item[(b)] With strategy profile $\vec{\sigma}^s$, the expected utility of the receiver in $\Gamma(X)$ and $\Gamma(X')$ is identical.
\end{itemize}
\end{lemma}
\begin{proof}
Recall that $\vec{\sigma}^s$ is a strategy profile in which, for each possible signal $x$, the sender sends its preferred action and then the receiver plays whatever is sent by the sender. This means that, if $\vec{\sigma}$ is incentive-compatible for the receiver in $\Gamma(X)$, then we can merge all signals in which the sender prefers $0$, and also merge all signals in which the sender prefers $1$. More precisely, given filter $X$, let $X_0$ and $X_1$ be the sets of signals in which the sender prefers $0$ and $1$ respectively. Consider a filter $X'$ that sends signal $0$ whenever $X$ would send a signal in $X_0$, and sends signal $1$ whenever $X$ would send a signal in $X_1$. In $\Gamma(X')$, by construction, both the sender and the receiver prefer action $0$ on signal $0$ and action $1$ on signal $1$. This means that $\vec{\sigma}^s$ is a Nash equilibrium of $\Gamma(X')$. Moreover, again by construction, the expected utility of the receiver (and the sender) when playing $\vec{\sigma}^s$ in $\Gamma(X')$ is identical to the one they'd get in $\Gamma(X)$.
\end{proof}

Recall that binary filters $X$ can be described by a function $X_* : \Omega \longrightarrow [0,1]$ that maps each state $\omega \in \Omega$ to the probability that $X(\omega)$ assigns to $0$. 
Because of Lemma~\ref{lemma:char}, our aim is to find which values in $[0,1]$ we should assign to each element in $\Omega$. The following lemmas characterizes these values.

\begin{lemma}\label{lem:char-values1}
There exists a filter $X_r^\Gamma$ that maximizes the utility for the receiver such that
\begin{itemize}
\item[(a)] $(X_r^\Gamma)_*(\omega) = 1$ for all $\omega \in \Omega_0^0$.
\item[(b)] $(X_r^\Gamma)_*(\omega) = 0$ for all $\omega \in \Omega_1^1$.
\item[(c)] At most one state $\omega \in \Omega_0^1 \cup \Omega_1^0$ satisfies that $(X_r^\Gamma)_*(\omega) \not \in \{0,1\}$. 
\end{itemize}
\end{lemma}

\begin{proof}
Given a binary filter $X$, if we increase $X_*(\omega_i)$ by $\epsilon$, the expected utility of the receiver increases by $\epsilon \cdot  p(\omega_i) \cdot d_i^r$. This means that, if $d_i^s \ge 0$ and $d_i^r \ge 0$, setting $X_*(\omega_i)$ to $1$ increases the receiver's utility and preserves the incentive-compatibility constraints (see Equation~\ref{eq:ic-s}). Analogously, the same happens by setting $X_*(\omega_i)$ to $0$ when $d_i^s \le 0$ and $d_i^r \le 0$. This proves (a) and (b).

To prove (c) suppose that there exist two states $\omega_i$ and $\omega_j$ such that $X_*^{OPT}(\omega_i), X_*^{OPT}(\omega_j) \not \in \{0,1\}$. Then, because of (a) and (b) we can assume without loss of generality that $\omega_i, \omega_j \in \Omega_0^1 \cup \Omega_1^0$. 
Therefore, if we increase $X_*(\omega_i)$ by $\epsilon \cdot d_j^s$ and decrease $X_*(\omega_j)$ by $\epsilon \cdot d_j^s$, we'd have that the sender's utility and incentive-compatibility constraints remain unchanged, but the receiver's utility increases by $\epsilon \cdot  (d_j^s d_i^r - d_i^s d_j^r)$ (note that this value can be negative). This means that if we choose an $\epsilon$ that is small enough and is of the same sign as $d_j^s d_i^r - d_i^s d_j^r$, not only the expected utility of the receiver increases, but also the values $\sum_i p(\omega_i) \cdot d_i^r \cdot X_*(\omega_i)$ and $\sum_i p(\omega_i) \cdot d_i^r \cdot (1 - X_*(\omega_i))$ increase and decrease respectively. This contradicts the fact that the filter $X_r^\Gamma$ is optimal for the receiver.
\end{proof}

Lemma~\ref{lem:char-values1} shows that we can assign $(X_r^\Gamma)_*(\omega) = 1$ for all $\omega \in \Omega_0^0$, $(X_r^\Gamma)_*(\omega) = 0$ for all $\omega \in \Omega_0^0$, and  either probability $1$ or probability $0$ to all but at most one of the remaining states. Ideally, we would like to assign probability $1$ to all states in $\Omega_0^1$ and probability $0$ to all states in $\Omega_1^0$, since this guarantees that the receiver gets the maximum possible utility. However, this may not always be incentive-compatible for the sender, which means that we may have to assign to some of the states the probability that the sender prefers, as opposed to the probability that the receiver prefers. The following lemma characterizes these states.

\begin{lemma}\label{lem:char-values2}
Given an information transmission game $\Gamma^{F} = (A, \Omega, p, u)$, let $\omega^1, \omega^2, \ldots, \omega^{\ell}$ be the states in $\Omega_0^1 \cup \Omega^0_1$ sorted by $\frac{u_r(\omega, 0) - u_r(\omega, 1)}{u_s(\omega, 1) - u_s(\omega, 0)}$. Then, there exists a binary filter $X_r^\Gamma$ that is optimal for the receiver and a value $\ell' \le \ell$ such that:
\begin{itemize}
\item[(a)] $(X_r^\Gamma)_*(\omega) = 1$ for all $\omega \in \Omega_0^0$.
\item[(b)] $(X_r^\Gamma)_*(\omega) = 0$ for all $\omega \in \Omega_1^1$.
\item[(c)] For $i < \ell'$, $(X_r^\Gamma)_*(\omega^i) =  0$ if $\omega \in \Omega_1^0$ and $(X_r^\Gamma)_*(\omega^i) =  1$ if $\omega \in \Omega_0^1$.
\item[(d)] For $i > \ell'$, $(X_r^\Gamma)_*(\omega^i) =  0$ if $\omega \in \Omega_0^1$ and $(X_r^\Gamma)_*(\omega^i) =  0$ if $\omega \in \Omega_1^0$.
\end{itemize}
\end{lemma}

Lemma~\ref{lem:char-values2} states that, if we sort the states in which the  sender and the receiver have different preferences by the ratio between how much the receiver prefers $0$ with respect to $1$ and how much the sender prefers $1$ with respect to $0$ (note that these ratios are always positive), we get that we can split these states into two contiguous blocks in which, in the first block, we assign them the probability that is most convenient for the sender, and for the second block we assign them the probability that is most convenient for the receiver. The proof follows the lines of that of Lemma~\ref{lem:char-values1} part (c).

\begin{proof}[Proof of Lemma~\ref{lem:char-values2}]
Given two states $\omega, \omega' \in \Omega_0^1 \cup \Omega^0_1$, we say that $\omega < \omega'$ if $\frac{u_r(\omega, 0) - u_r(\omega, 1)}{u_s(\omega, 1) - u_s(\omega, 0)} < \frac{u_r(\omega', 0) - u_r(\omega', 1)}{u_s(\omega', 1) - u_s(\omega', 0)}$. Using the same argument as in the Proof of Lemma~\ref{lem:char-values1}, we can assume w.l.o.g. that (a) and (b) are satisfied. Thus, it only remains to show (c) and (d), which follow from the fact that there exists an there exists an optimal binary filter $X_r^\Gamma$ for the receiver such that:
\begin{itemize}
\item[(s1)] If $\omega, \omega' \in \Omega_1^0$, $\omega < \omega'$ and $(X_r^\Gamma)_*(\omega) > 0$, then $(X_r^\Gamma)_*(\omega') = 1$.
\item[(s2)] If $\omega, \omega' \in \Omega_0^1$, $\omega < \omega'$ and $(X_r^\Gamma)_*(\omega) < 1$, then $(X_r^\Gamma)_*(\omega') = 0$.
\item[(s3)] If $\omega \in \Omega_1^0$, $\omega' \in \Omega_0^1$, $\omega < \omega'$ and $(X_r^\Gamma)_*(\omega) > 0$, then $(X_r^\Gamma)_*(\omega') = 0$.
\item[(s4)] If $\omega \in \Omega_0^1$, $\omega' \in \Omega_1^0$, $\omega < \omega'$ and $(X_r^\Gamma)_*(\omega) < 1$, then $(X_r^\Gamma)_*(\omega') = 1$.
\end{itemize}

We will prove (s1) and (s3), the proofs of (s2) and (s4) are analogous. 

\textbf{Proof of (s1):} Suppose that there exist $\omega, \omega' \in \Omega_1^0$ such that $\omega < \omega'$ and $(X_r^\Gamma)_*(\omega) > 0$, but $(X_r^\Gamma)_*(\omega') < 1$. If this happens, we can 
set $$
\begin{array}{l}
(X_r^\Gamma)_*(\omega) \longleftarrow (X_r^\Gamma)_*(\omega) + \epsilon (u_s(\omega', 0) - u_s(\omega', 1)) \\ 
(X_r^\Gamma)_*(\omega') \longleftarrow (X_r^\Gamma)_*(\omega) + \epsilon (u_s(\omega, 1) - u_s(\omega, 0)).
\end{array}$$
Since $u_s(\omega', 0) - u_s(\omega', 1)) < 0$ and $u_s(\omega, 0) - u_s(\omega, 1) < 0$, there exists some $\epsilon > 0$ that is small enough so that $X^{OPT}_*(\omega)$ and $(X_r^\Gamma)_*(\omega)$ stay between $0$ and $1$.
Moreover, as in the proof of Lemma~\ref{lem:char-values1} part (c), if we perform this change the sender's utility remains unchanged, but the receiver's utility increases by $$\Delta := \epsilon((u_s(\omega', 0) - u_s(\omega', 1))(u_r(\omega, 0) - u_r(\omega, 1)) + (u_s(\omega, 1) - u_s(\omega, 0))(u_r(\omega', 0) - u_r(\omega', 1))).$$

By assumption, we have that $$\frac{u_r(\omega, 0) - u_r(\omega, 1)}{u_s(\omega, 1) - u_s(\omega, 0)} < \frac{u_r(\omega', 0) - u_r(\omega', 1)}{u_s(\omega', 1) - u_s(\omega', 0)}$$ which, together with the fact that $u_s(\omega, 1) - u_s(\omega, 0) > 0$ and $u_s(\omega', 1) - u_s(\omega', 0)> 0$, implies that $\Delta \ge 0$ whenever $\epsilon > 0$.

\textbf{Proof of (s3):} The proof is almost identical to that of (s1). Suppose that there exist $\omega \in \Omega_1^0$, $\omega' \in \Omega_0^1$ such that $\omega < \omega'$, $(X_r^\Gamma)_*(\omega) > 0$ and $(X_r^\Gamma)_*(\omega')> 0$. In this case, we can again set 
$$
\begin{array}{l}
(X_r^\Gamma)_*(\omega) \longleftarrow (X_r^\Gamma)_*(\omega) + \epsilon (u_s(\omega', 1) - u_s(\omega', 0)) \\ 
(X_r^\Gamma)_*(\omega') \longleftarrow (X_r^\Gamma)_*(\omega) + \epsilon (u_s(\omega, 0) - u_s(\omega, 1)).
\end{array}$$
The only difference with the proof of (s1) is that, in this case, $u_s(\omega, 0) - u_s(\omega, 1) < 0$ and  $u_s(\omega', 0) - u_s(\omega', 1) > 0$. Again, there exists a sufficiently small $\epsilon > 0$ such that $(X_r^\Gamma)_*(\omega)$ and $(X_r^\Gamma)_*(\omega')$ remain between $0$ and $1$. Moreover, a straightforward computation shows that the sender's utility remains unchanged, but that the receiver's utility changes by $$\Delta := \epsilon((u_s(\omega', 1) - u_s(\omega', 0))(u_r(\omega, 0) - u_r(\omega, 1)) + (u_s(\omega, 0) - u_s(\omega, 1))(u_r(\omega', 0) - u_r(\omega', 1))).$$

Using that $$\frac{u_r(\omega, 0) - u_r(\omega, 1)}{u_s(\omega, 1) - u_s(\omega, 0)} < \frac{u_r(\omega', 1) - u_r(\omega', 0)}{u_s(\omega', 0) - u_s(\omega', 1)},$$ $u_s(\omega, 1) - u_s(\omega, 0) > 0$, and $u_s(\omega', 0) - u_s(\omega', 1)> 0$, we get that $\Delta \ge 0$ whenever $\epsilon > 0$.
\end{proof}

Lemma~\ref{lem:char-values2} shows that the algorithm provided in Section~\ref{sec:algorithm} outputs the correct solution as long as it computes the right value of $q$ in Step 5. The next lemma shows that this value is precisely the one that the algorithm finds. This completes the proof of correctness.

\begin{lemma}
Let $X_R^\Gamma$ be the filter that assigns $(X_R^\Gamma)_*(\omega) = 1$ for $\omega \in \Omega_0^0 \cup \Omega_1^0$ and $(X_R^\Gamma)_*(\omega) = 0$ for $\omega \in \Omega_1^1 \cup \Omega_0^1$. If $X_R^\Gamma$ is incentive-compatible for the sender, $(X_R^\Gamma)$ is the optimal filter for the receiver that is incentive-compatible for both players. Otherwise, all filters $X$ that are incentive-compatible for both players and are optimal for the receiver satisfy at least one of the following equations:

$$
\begin{array}{ll}
\sum_i p(\omega_i) \cdot d_i^s \cdot X_*(\omega_i) = 0\\
\sum_i p(\omega_i) \cdot d_i^s \cdot (1 - X_*(\omega_i)) = 0.
\end{array}
$$

\end{lemma}

\begin{proof}
Filter $X_R^\Gamma$ gives the maximum utility to the sender of all possible filters. Therefore, if it is incentive-compatible for the sender, it is the optimal for the receiver. Suppose instead that $X_R^\Gamma$ is not incentive-compatible for the sender but that an optimal filter $X_r^\Gamma$ satisfies 
\begin{equation}\label{eq:lemma-3}
\begin{array}{ll}
\sum_i p(\omega_i) \cdot d_i^s \cdot X_*(\omega_i) > 0\\
\sum_i p(\omega_i) \cdot d_i^s \cdot (1 - X_*(\omega_i)) < 0.
\end{array}
\end{equation}

Since $X_R^\Gamma$ is not IC for the sender, there exists a state $\omega \in \Omega_{0,1}$ such that $(X_r^\Gamma)_* > 0$ or a state $\omega \in \Omega_{1,0}$ such that $(X_r^\Gamma)_* < 1$. In the first case, we can decrease $(X_r^\Gamma)_*$ by a small value $\epsilon > 0$ such that Equation~\ref{eq:lemma-3} is still satisfied. By doing this, we increase the receiver's expected utility while obtaining a new filter that is still incentive-compatible for both players. The latter case is analogous except that we increase $(X_r^\Gamma)_*$ instead of decreasing it.
\end{proof}

\section{Constructing a Sender-Optimal Filter}\label{sec:sender-algorithm}

Constructing a sender-optimal filter $X_s^\Gamma$ in an information transmission game is analogous the one given in Section~\ref{sec:algorithm} but ``reversing'' the roles of the sender and the receiver. More precisely, the construction is as follows.

\begin{enumerate}
\item \textbf{Step 1:} Set $(X_s^\Gamma)_*(\omega) = 1$ for all $\omega \in \Omega^0$.
\item \textbf{Step 2:} Set $(X_s^\Gamma)_*(\omega) = 0$ for all $\omega \in \Omega^1$.
\item \textbf{Step 3:} Sort all states $\omega \in \Omega_{0}^1 \cup \Omega_1^0$ according to the value $\frac{u_s(\omega, 0) - u_s(\omega, 1)}{u_r(\omega, 1) - u_r(\omega, 0)}$. Let $\omega^1, \omega^2, \ldots, \omega^{k'}$ be the resulting list.
\item \textbf{Step 4:} Set $(X_s^\Gamma)_*(\omega) = 1$ for all $\omega \in \Omega_0^1$ and $(X_s^\Gamma)_*(\omega) = 0$ for all $\omega \in \Omega_1^0$. If the resulting filter is incentive-compatible for the receiver, output $X_s^\Gamma$.
\item \textbf{Step 5:} For $i = 1, 2, \ldots, k'$ do the following. If $(X_s^\Gamma)_*(\omega^i)$ is set to $1$ (resp., to $0$), set it to $0$ (resp., to 1). If the resulting filter is incentive-compatible for the receiver, find the maximum (resp., the minimum) $q$ such that setting $(X_s^\Gamma)_*(\omega^i)$ to $q$ is incentive-compatible for the receiver. If the resulting filter  is incentive-compatible for the sender, output $X_s^\Gamma$, otherwise output the constant filter $(X_s^\Gamma)_* \equiv 0$.
\end{enumerate}

The proof of correctness is analogous to the one shown in Section~\ref{sec:correctness}.

\section{Proof of Theorem~\ref{thm:two-senders}}\label{sec:proof-2-senders}

The proof of Theorem~\ref{thm:two-senders} follows from the following result of Arieli et al.~\cite{arieli2023mediated} describing the characterization of the best Nash equilibrium in information aggregation games with two senders, one receiver, binary actions and a mediator.

\begin{theorem}[\cite{arieli2023mediated}]\label{thm:mediator}
Let $\Gamma = \{S, A, \Omega, p, u\}$ be an information aggregation game with $S = \{s_1, s_2\}$ and $A = \{0,1\}$ in which the players can communicate with a third-party mediator. Given $i,j,k \in \{0,1\}$, let $\Omega_{i,j}^k$ be the set of states in which $s_1$ prefers action $i$, $s_2$ prefers action $j$, and the receiver prefers action $k$. Let also $\Omega_{i,j} := \Omega_{i,j}^0 \cup \Omega_{i,j}^1$. Consider the following six maps $M_1, M_2, M_3, M_4, M_5, M_6$ from $\Omega$ to $[0,1]$:

\begin{itemize}
    \item [(a)] $M_1(\omega) = 1$ for all $\omega \in \Omega_{0,0}^0$ and $M_1(\omega) = 0$ otherwise.
    \item [(b)] $M_2(\omega) = 0$ for all $\omega \in \Omega_{1,1}^1$ and $M_2(\omega) = 1$ otherwise.
    \item [(c)] $M_3(\omega) = 1$ for all $\omega \in \Omega_{0,0} \cup \Omega_{0,1}$ and $M_3(\omega) = 0$ otherwise.
    \item [(d)] $M_4(\omega) = 1$ for all $\omega \in \Omega_{0,0} \cup \Omega_{1,0}$ and $M_4(\omega) = 0$ otherwise.
    \item [(e)] $M_5(\omega) = 1$ for all $\omega \in \Omega$.
    \item [(f)] $M_6(\omega) = 0$ for all $\omega \in \Omega$.
\end{itemize}

Then, there exists an $i \in \{1,2,3,4,5,6\}$ and a Nash equilibrium of $\Gamma$ that maximizes the receiver's utility in which the function that maps each state to the probability that the receiver ends up playing $0$ is equal to $M_i$.
\end{theorem}

Intuitively, Theorem~\ref{thm:mediator} states that if there were no filters and the two senders had access to a third-party mediator, there exists a Nash equilibrium that is optimal for the receiver in which the either (a) the receiver only plays $0$ if all three players prefer $0$, (b) the sender plays $1$ if and only if all three players prefer $1$, (c) the receiver always plays what the first sender prefers, (d) the receiver always plays what the second sender prefers, (e) the receiver plays always 0, or (f) the receiver plays always 1. In the setting described in Section~\ref{sec:model}, players have no access to a mediator. However, we can easily check if the outcomes described in (a), (b), (c), (d), (e) or (f) are incentive-compatible for the receiver (i.e., if the receiver gets more utility with the outcome than when playing with no information), there exists a strategy in $\Gamma$ (without the mediator) that implements these outcomes.

For instance, suppose that we want to implement the outcome described in (a). Consider the strategy profile $\vec{\sigma}^{(0,0)}$ in which each sender sends $0$ to the receiver if and only if the realized state is in $\Omega_{0,0}^0$, and the receiver plays $0$ only if both signals are $0$. It is easy to check that this is incentive-compatible for the senders: if the realized state is indeed in $\Omega_{0,0}^0$, it is a best response for both to send $0$ since it guarantees that the receiver will play $0$. Moreover, if the realized state is not in $\Omega_{0,0}^0$, none of the senders gets any additional utility by defecting from the main strategy since the other sender will always send a non-zero signal (which implies that the receiver will play $1$). To implement (c), consider the strategy profile $\vec{\sigma}^{s_1}$ in which both senders send a binary signal $m \in \{0,1\}$ that is equal to $0$ if and only if the realized state is in $\Omega_{0,0} \cup \Omega_{0,1}$, and the receiver plays $0$ if and only if the signal from the first sender is $0$. Again, it is straightforward to check that this is incentive-compatible for the senders: it is a best-response for sender $1$ to send its preference, while it doesn't matter what sender $2$ sends since it will be ignored. To implement (e) (resp., (f)) consider the strategy profile $\vec{\sigma}^0$ (resp., $\vec{\sigma}^1$) in which the senders send signal $0$ regardless of the realized state and the receiver always plays $0$ (resp., $1$). It is easy to check that $\vec{\sigma}^0$ (resp., $\vec{\sigma}^1$) are incentive-compatible for the senders. 
Implementing the outcomes in (b) and (d) is analogous to implementing the outcomes in (a) and (c).
Since all outcomes that are implementable without a mediator are also implementable with a mediator, this implies the following proposition:

\begin{proposition}\label{prop:two-senders-optimal}
Let $\Gamma = \{S, A, \Omega, p, u\}$ be an information aggregation game with $S = \{s_1, s_2\}$ and $A = \{0,1\}$. Then, either $\vec{\sigma}^0$, $\vec{\sigma}^1$, $\vec{\sigma}^{(0,0)}$, $\vec{\sigma}^{(1,1)}$, $\vec{\sigma}^{s_1}$ or $\vec{\sigma}^{s_2}$ is a Nash equilibrium that is optimal for the receiver.
\end{proposition}

Proposition~\ref{prop:two-senders-optimal} implies that we can break the problem of finding the receiver-optimal filter in an information aggregation game into four sub-problems:
\begin{itemize}
    \item[(a)] Finding a filter $(X_r^\Gamma)_{(0,0)}$ such that $\vec{\sigma}^{(0,0)}$ gives the maximal utility for the receiver in $\Gamma((X_r^\Gamma)_{(0,0)})$.
    \item[(b)] Finding a filter $(X_r^\Gamma)_{(1,1)}$ such that $\vec{\sigma}^{(1,1)}$ gives the maximal utility for the receiver in $\Gamma((X_r^\Gamma)_{(1,1)})$.
    \item[(c)] Finding a filter $(X_r^\Gamma)_{s_1}$ such that $\vec{\sigma}^{s_1}$ gives the maximal utility for the receiver in $\Gamma((X_r^\Gamma)_{s_1})$.
    \item[(d)] Finding a filter $(X_r^\Gamma)_{s_2}$ such that $\vec{\sigma}^{s_2}$ gives the maximal utility for the receiver in $\Gamma((X_r^\Gamma)_{s_2})$.
\end{itemize}

Note that these filters might not always exist (for instance, $(X_r^\Gamma)_{(0,0)}$ does not exist when the receiver always prefers $0$ and the senders always prefer $1$). Moreover, we are not including the optimal filters for $\vec{\sigma}^0$ and $\vec{\sigma}^1$ since all filters give the same utility with these strategies. Finding $(X_r^\Gamma)_{s_1}$ and $(X_r^\Gamma)_{s_2}$ (whenever they exist) reduce to the case of one sender (by ignoring $s_2$ and $s_1$, respectively), and thus can be solved with the $O(k \log k)$ algorithm presented in Section~\ref{sec:algorithm}. We next show how to compute $(X_r^\Gamma)_{(0,0)}$ using a linear program with $k$ variables and $2k+2$ constraints. Computing $(X_r^\Gamma)_{(1,1)}$ is analogous.

Suppose that $X$ is a filter that maximizes the receiver utility in $\Gamma(X)$ when playing strategy $\vec{\sigma}^{(0,0)}$. Denote by $M_0$ the set of signals of $X$ in which both senders and the receiver prefer $0$ to $1$. Consider a filter $X'$ that samples a signal $m$ according to $X$ and does the following: if $m \in M_0$, it sends signal $0$. Otherwise, it sends signal $1$. It is straightforward to check that if $\vec{\sigma}^{(0,0)}$ is a Nash equilibrium in $\Gamma^F_X$, it is also a Nash equilibrium in $\Gamma(X')$. Moreover, players get the same utilities with $X$ and $X'$, which means that $X'$ also maximizes the receiver's utility. This implies that we can restrict our search to binary filters in which all players prefer action $0$ on signal $0$ and the receiver prefers action $1$ on signal $1$. 

\commentout{Our aim then is finding a binary filter $X$ such that both senders and the receiver prefer action $0$ on signal $0$, the receiver prefers action $1$ on signal $1$, and such that the receiver gets the maximal possible utility by playing $0$ on signal $0$ and $1$ on signal $1$.}
Let $\omega_1, \omega_2, \ldots, \omega_k$ be the elements of $\Omega$ and define $A_i := u(s_1, \omega_i, 0) - u(s_1, \omega_i, 1)$, $B_i := u(s_2, \omega_i, 0) - u(s_2, \omega_i, 1)$ and $C_i := u(r, \omega_i, 0) - u(r, \omega_i, 1)$. A binary filter $X$ over $\Omega$ can be described by a sequence of $k$ real numbers $x_1, \ldots, x_k$ between $0$ and $1$ such that $x_i$ denotes the probability that the filter sends signal $0$ on state $\omega_i$. The condition that both senders prefer action $0$ on signal $0$ translates into $\sum_{i = 1}^k A_i x_i \ge 0$ and $\sum_{i = 1}^k B_i x_i \ge 0$.
\commentout{
$$\begin{array}{c}
    \sum_{i = 1}^k A_i x_i \ge 0\\
    \sum_{i = 1}^k B_i x_i \ge 0.\\ 
    \end{array}$$
}

Moreover, the utility of the receiver with filter $X$ and strategy $\vec{\sigma}^{(0,0)}$ is given by $\sum_{i = 1}^k (x_i \cdot u(r, \omega_i, 0) + (1 - x_i) \cdot u(r, \omega_i, 1))$. This sum can be rearranged into $\sum_{i = 1}^k u(r, \omega_i, 1) + \sum_{i = 1}^k C_i x_i$. Therefore, finding a filter $(X_r^\Gamma)_{(0,0)}$ such that $\vec{\sigma}^{(0,0)}$ gives the maximal utility for the receiver in $\Gamma((X_r^\Gamma)_{(0,0)})$ reduces to solving the following linear programming instance:

$$ \begin{array}{c}
\max \sum_{i = 1}^k C_i x_i\\
\vspace{1mm}
\sum_{i = 1}^k A_i x_i \ge 0\\
\vspace{1mm}
\sum_{i = 1}^k B_i x_i \ge 0\\
\vspace{1mm}
0 \le x_i \le 1 \quad \forall i \in [k].
\end{array} $$

Note that, even though we are maximizing the receiver's utility, it may be the case that $\vec{\sigma}^{(0,0)}$ is not incentive-compatible for the receiver in $\Gamma((X_r^\Gamma)_{(0,0)})$. For instance, the receiver might prefer playing $0$ whenever the senders send $1$. If such a thing happens (which can be tested in linear time), it means  that there is no filter satisfying the desired conditions.

\section{Conclusion}\label{sec:conclusion}

Garbling the information of the sender can sometimes increase the players' utilities in information transmission and information aggregation games, and the optimal garbling for both sender and receiver can be computed with efficient algorithms. However, this paper leaves two important open questions. First if the algorithm presented in Section~\ref{sec:algorithm} can be generalized to three or more actions and, second, if we can construct an efficient algorithm that outputs the best filter for one of the senders in an information aggregation game with multiple senders.

\bibliographystyle{eptcs}
\bibliography{bibfile}

\begin{thebibliography}{10}
\providecommand{\bibitemdeclare}[2]{}
\providecommand{\surnamestart}{}
\providecommand{\surnameend}{}
\providecommand{\urlprefix}{Available at }
\providecommand{\url}[1]{\texttt{#1}}
\providecommand{\href}[2]{\texttt{#2}}
\providecommand{\urlalt}[2]{\href{#1}{#2}}
\providecommand{\doi}[1]{doi:\urlalt{https://doi.org/#1}{#1}}
\providecommand{\eprint}[1]{arXiv:\urlalt{https://arxiv.org/abs/#1}{#1}}
\providecommand{\bibinfo}[2]{#2}

\bibitemdeclare{inproceedings}{arieli2023mediated}
\bibitem{arieli2023mediated}
\bibinfo{author}{Itai \surnamestart Arieli\surnameend}, \bibinfo{author}{Ivan
  \surnamestart Geffner\surnameend} \& \bibinfo{author}{Moshe \surnamestart
  Tennenholtz\surnameend} (\bibinfo{year}{2023}):
  \emph{\bibinfo{title}{Mediated cheap talk design}}.
\newblock In: {\slshape \bibinfo{booktitle}{Proceedings of the AAAI Conference
  on Artificial Intelligence}}, \bibinfo{volume}{37}, pp.
  \bibinfo{pages}{5456--5463}, \doi{10.1609/aaai.v37i5.25678}.

\bibitemdeclare{inproceedings}{arieli2023resilient}
\bibitem{arieli2023resilient}
\bibinfo{author}{Itai \surnamestart Arieli\surnameend}, \bibinfo{author}{Ivan
  \surnamestart Geffner\surnameend} \& \bibinfo{author}{Moshe \surnamestart
  Tennenholtz\surnameend} (\bibinfo{year}{2023}):
  \emph{\bibinfo{title}{Resilient Information Aggregation}}.
\newblock In: {\slshape \bibinfo{booktitle}{Proceedings of the 2023 Conference
  on Theoretical Aspects of Rationality and Knowledge (TARK)}}, {\slshape
  \bibinfo{series}{{EPTCS}}} \bibinfo{volume}{379}, pp.
  \bibinfo{pages}{31--45}, \doi{10.4204/EPTCS.379.6}.

\bibitemdeclare{article}{bergemann2015limits}
\bibitem{bergemann2015limits}
\bibinfo{author}{Dirk \surnamestart Bergemann\surnameend},
  \bibinfo{author}{Benjamin \surnamestart Brooks\surnameend} \&
  \bibinfo{author}{Stephen \surnamestart Morris\surnameend}
  (\bibinfo{year}{2015}): \emph{\bibinfo{title}{The limits of price
  discrimination}}.
\newblock {\slshape \bibinfo{journal}{American Economic Review}}
  \bibinfo{volume}{105}(\bibinfo{number}{3}), pp. \bibinfo{pages}{921--957},
  \doi{10.2139/ssrn.2501403}.

\bibitemdeclare{article}{blume2013language}
\bibitem{blume2013language}
\bibinfo{author}{Andreas \surnamestart Blume\surnameend} \&
  \bibinfo{author}{Oliver \surnamestart Board\surnameend}
  (\bibinfo{year}{2013}): \emph{\bibinfo{title}{Language barriers}}.
\newblock {\slshape \bibinfo{journal}{Econometrica}}
  \bibinfo{volume}{81}(\bibinfo{number}{2}), pp. \bibinfo{pages}{781--812},
  \doi{10.3982/ecta9183}.

\bibitemdeclare{inproceedings}{corrao2023mediated}
\bibitem{corrao2023mediated}
\bibinfo{author}{Roberto \surnamestart Corrao\surnameend} \&
  \bibinfo{author}{Yifan \surnamestart Dai\surnameend} (\bibinfo{year}{2023}):
  \emph{\bibinfo{title}{Mediated Communication with Transparent Motives}}.
\newblock In: {\slshape \bibinfo{booktitle}{Proceedings of the 24th ACM
  Conference on Economics and Computation}}, pp. \bibinfo{pages}{489--489},
  \doi{10.1145/3580507.3597808}.

\bibitemdeclare{article}{de2003game}
\bibitem{de2003game}
\bibinfo{author}{Kris \surnamestart De~Jaegher\surnameend}
  (\bibinfo{year}{2003}): \emph{\bibinfo{title}{A game-theoretic rationale for
  vagueness}}.
\newblock {\slshape \bibinfo{journal}{Linguistics and Philosophy}}
  \bibinfo{volume}{26}(\bibinfo{number}{5}), pp. \bibinfo{pages}{637--659},
  \doi{10.1023/A:1025853728992}.

\bibitemdeclare{article}{glukhov2023llm}
\bibitem{glukhov2023llm}
\bibinfo{author}{David \surnamestart Glukhov\surnameend}, \bibinfo{author}{Ilia
  \surnamestart Shumailov\surnameend}, \bibinfo{author}{Yarin \surnamestart
  Gal\surnameend}, \bibinfo{author}{Nicolas \surnamestart Papernot\surnameend}
  \& \bibinfo{author}{Vardan \surnamestart Papyan\surnameend}
  (\bibinfo{year}{2023}): \emph{\bibinfo{title}{Llm censorship: A machine
  learning challenge or a computer security problem?}}
\newblock {\slshape \bibinfo{journal}{arXiv preprint arXiv:2307.10719}},
  \doi{10.48550/arXiv.2307.10719}.

\bibitemdeclare{article}{hagenbach2020cheap}
\bibitem{hagenbach2020cheap}
\bibinfo{author}{Jeanne \surnamestart Hagenbach\surnameend} \&
  \bibinfo{author}{Fr{\'e}d{\'e}ric \surnamestart Koessler\surnameend}
  (\bibinfo{year}{2020}): \emph{\bibinfo{title}{Cheap talk with coarse
  understanding}}.
\newblock {\slshape \bibinfo{journal}{Games and Economic Behavior}}
  \bibinfo{volume}{124}, pp. \bibinfo{pages}{105--121},
  \doi{10.1016/j.geb.2020.07.015}.

\bibitemdeclare{article}{ichihashi2019limiting}
\bibitem{ichihashi2019limiting}
\bibinfo{author}{Shota \surnamestart Ichihashi\surnameend}
  (\bibinfo{year}{2019}): \emph{\bibinfo{title}{Limiting Sender's information
  in Bayesian persuasion}}.
\newblock {\slshape \bibinfo{journal}{Games and Economic Behavior}}
  \bibinfo{volume}{117}, pp. \bibinfo{pages}{276--288},
  \doi{10.2139/ssrn.3233072}.

\bibitemdeclare{article}{jager2011voronoi}
\bibitem{jager2011voronoi}
\bibinfo{author}{Gerhard \surnamestart J{\"a}ger\surnameend},
  \bibinfo{author}{Lars~P \surnamestart Metzger\surnameend} \&
  \bibinfo{author}{Frank \surnamestart Riedel\surnameend}
  (\bibinfo{year}{2011}): \emph{\bibinfo{title}{Voronoi languages: Equilibria
  in cheap-talk games with high-dimensional types and few signals}}.
\newblock {\slshape \bibinfo{journal}{Games and economic behavior}}
  \bibinfo{volume}{73}(\bibinfo{number}{2}), pp. \bibinfo{pages}{517--537},
  \doi{10.1016/j.geb.2011.03.008}.

\bibitemdeclare{article}{kamenica2011bayesian}
\bibitem{kamenica2011bayesian}
\bibinfo{author}{Emir \surnamestart Kamenica\surnameend} \&
  \bibinfo{author}{Matthew \surnamestart Gentzkow\surnameend}
  (\bibinfo{year}{2011}): \emph{\bibinfo{title}{Bayesian persuasion}}.
\newblock {\slshape \bibinfo{journal}{American Economic Review}}
  \bibinfo{volume}{101}(\bibinfo{number}{6}), pp. \bibinfo{pages}{2590--2615},
  \doi{10.3386/w15540}.

\bibitemdeclare{article}{lipnowski2020cheap}
\bibitem{lipnowski2020cheap}
\bibinfo{author}{Elliot \surnamestart Lipnowski\surnameend} \&
  \bibinfo{author}{Doron \surnamestart Ravid\surnameend}
  (\bibinfo{year}{2020}): \emph{\bibinfo{title}{Cheap talk with transparent
  motives}}.
\newblock {\slshape \bibinfo{journal}{Econometrica}}
  \bibinfo{volume}{88}(\bibinfo{number}{4}), pp. \bibinfo{pages}{1631--1660},
  \doi{10.3982/ecta15674}.

\bibitemdeclare{article}{liu2024protecting}
\bibitem{liu2024protecting}
\bibinfo{author}{Zichuan \surnamestart Liu\surnameend}, \bibinfo{author}{Zefan
  \surnamestart Wang\surnameend}, \bibinfo{author}{Linjie \surnamestart
  Xu\surnameend}, \bibinfo{author}{Jinyu \surnamestart Wang\surnameend},
  \bibinfo{author}{Lei \surnamestart Song\surnameend},
  \bibinfo{author}{Tianchun \surnamestart Wang\surnameend},
  \bibinfo{author}{Chunlin \surnamestart Chen\surnameend}, \bibinfo{author}{Wei
  \surnamestart Cheng\surnameend} \& \bibinfo{author}{Jiang \surnamestart
  Bian\surnameend} (\bibinfo{year}{2024}): \emph{\bibinfo{title}{Protecting
  your llms with information bottleneck}}.
\newblock {\slshape \bibinfo{journal}{Advances in Neural Information Processing
  Systems}} \bibinfo{volume}{37}, pp. \bibinfo{pages}{29723--29753},
  \doi{10.48550/arXiv.2404.13968}.

\bibitemdeclare{article}{patil2023can}
\bibitem{patil2023can}
\bibinfo{author}{Vaidehi \surnamestart Patil\surnameend},
  \bibinfo{author}{Peter \surnamestart Hase\surnameend} \&
  \bibinfo{author}{Mohit \surnamestart Bansal\surnameend}
  (\bibinfo{year}{2023}): \emph{\bibinfo{title}{Can sensitive information be
  deleted from llms? objectives for defending against extraction attacks}}.
\newblock {\slshape \bibinfo{journal}{arXiv preprint arXiv:2309.17410}},
  \doi{10.48550/arXiv.2309.17410}.

\end{thebibliography}

\end{document}